\documentclass[11pt]{article}

\usepackage[colorlinks]{hyperref}
\usepackage{amsmath} 
\usepackage{amsthm} 
\usepackage{amssymb}	
\usepackage{graphicx} 
\usepackage{multicol} 
\usepackage{multirow}
\usepackage{color}
\usepackage[dvips,letterpaper,margin=1in,bottom=1in]{geometry}
\usepackage[capitalize,noabbrev]{cleveref}

\usepackage[utf8]{inputenc}
\usepackage[english]{babel}

\usepackage{diagbox}
\usepackage{mathtools}

\newtheorem{theorem}{Theorem}[section]

\newtheorem{lemma}[theorem]{Lemma}

\newtheorem{proposition}[theorem]{Proposition}

\newtheorem{remark}{Remark}[section]
\newtheorem{problem}{Problem}

\DeclarePairedDelimiter\rbra{\lparen}{\rparen}
\DeclarePairedDelimiter\sbra{\lbrack}{\rbrack}
\DeclarePairedDelimiter\cbra{\{}{\}}
\DeclarePairedDelimiter\abs{\lvert}{\rvert}
\DeclarePairedDelimiter\Abs{\lVert}{\rVert}
\DeclarePairedDelimiter\ceil{\lceil}{\rceil}

\DeclarePairedDelimiter\ket{\lvert}{\rangle}

\newcommand{\substr}[2] {\sbra*{#1 .. #2}}

\newcommand{\set}[2] {\left\{\, #1 \colon #2 \,\right\}}

\newcommand{\poly} {\operatorname{poly}}

\usepackage{latexsym}
\usepackage{CJK}

\usepackage{enumerate}

\usepackage{algorithm}
\usepackage{algpseudocode}

\usepackage{stmaryrd}
\usepackage{booktabs}

\usepackage{cite}
\usepackage{bbm}

\usepackage{tablefootnote}
\usepackage{threeparttable}
\usepackage{scrextend}

\usepackage{tikz}
\usetikzlibrary{positioning}

\begin{document}

\title{A Note on Quantum Divide and Conquer \\ for Minimal String Rotation\footnote{An extended abstract of an earlier version of this paper was presented at the 47th Quantum Information Technology Symposium (QIT47) in Yokohama, Japan, in 2022.}}
    \author{
        Qisheng Wang \footnote{Qisheng Wang is with the School of Informatics, University of Edinburgh, EH8 9AB Edinburgh, United Kingdom (e-mail: \href{mailto:QishengWang1994@gmail.com}{\nolinkurl{QishengWang1994@gmail.com}}). Part of the work was done when the author was with the Graduate School of Mathematics, Nagoya University, Nagoya 464-8602, Japan.}
    }
    \date{}
    \maketitle

\begin{abstract}
    Lexicographically minimal string rotation is a fundamental problem in string processing that has recently garnered significant attention in quantum computing. Near-optimal quantum algorithms have been proposed for solving this problem, utilizing a divide-and-conquer structure. In this note, we show that its quantum query complexity is $\sqrt{n} \cdot 2^{O\rbra*{\sqrt{\log n}}}$, improving the prior result of $\sqrt{n} \cdot 2^{\rbra*{\log n}^{1/2+\varepsilon}}$ due to \hyperlink{cite.AJ22}{Akmal and Jin (SODA 2022)}.
    Notably, this improvement is quasi-polylogarithmic, which is achieved by only logarithmic level-wise optimization using fault-tolerant quantum minimum finding.
\end{abstract}


\newpage

\tableofcontents

\newpage

\section{Introduction}

Lexicographically minimal string rotation (LMSR) is a fundamental problem on string processing (cf. \cite{Jeu93,CR94,Gus97,CHL07}) that has lots of applications in different areas, e.g., combinatorial mathematics \cite{LWW01,JWW16}, computational geometry \cite{IS89,Mae91}, graph theory \cite{CB81,GY03}, automata theory \cite{AZ08,Pup10}, and computational chemistry \cite{Shi79}. Since Booth \cite{Boo80} proposed the first linear-time algorithm for LMSR based on the KMP string matching algorithm \cite{KMP77}, a series of studies have been carried out successively on its sequential \cite{Shi81,Duv83,Cro92,IS94,Ryt03,BCN05} and parallel \cite{AIP87,IS92} algorithms.

Recently, Wang and Ying \cite{WY23} achieved a quantum speedup on LMSR based on an observation, called ``exclusion rule'' or ``Ricochet Property'' in the literature, borrowed from parallel string processing algorithms \cite{Vis91,IS92}.
Immediately after, Akmal and Jin \cite{AJ22} improved the exclusion rule given by \cite{WY23} and developed a quantum algorithm with a divide-and-conquer structure. The (worst-case) query complexity $\sqrt{n} \cdot 2^{\rbra*{\log n}^{1/2+\varepsilon}}$ of their quantum algorithm is near-optimal up to $n^{o\rbra*{1}}$ factors,\footnote{Although it was explicitly stated in \cite{AJ22} that the quantum query complexity of LMSR is $\sqrt{n} \cdot 2^{O\rbra*{\rbra*{\log n}^{2/3}}}$, a slightly more detailed analysis will show that it is $\sqrt{n} \cdot 2^{\rbra*{\log n}^{1/2+\varepsilon}}$ for any constant $\varepsilon > 0$.} improving the result of $O\rbra*{n^{3/4}}$ by \cite{WY23}.
Inspired by the structures of quantum algorithms in \cite{AGS19} and \cite{AJ22}, Childs, Kothari, Kovacs-Deak, Sundaram, and Wang \cite{CKKD+22} further investigated a general framework of quantum divide and conquer. Moreover, they applied their framework to various string problems and obtained near-optimal quantum algorithms for them, including recognizing regular languages, $k$-increasing subsequence, $k$-common subsequence, and LMSR.
Especially, they considered the decision version of LMSR, and showed that its quantum query complexity is $\widetilde O\rbra*{\sqrt{n}}$.\footnote{$\widetilde O\rbra*{\cdot}$ suppreses polylogarithmic factors.}
After the work of \cite{CKKD+22}, Allcock, Bao, Belovs, Lee, and Santha \cite{ABB+23} further studied the time efficiency of quantum divide and conquer.
However, the results of \cite{CKKD+22} and \cite{ABB+23} do not imply a better quantum query (and time) complexity of the function version of LMSR.

In this note, we further improve the quantum query complexity of (the function version of) LMSR from $\sqrt{n} \cdot 2^{\rbra*{\log n}^{1/2+\varepsilon}}$ by \cite{AJ22} to $\sqrt{n} \cdot 2^{O\rbra*{\sqrt{\log n}}}$.
Our approach is to improve the complexity in each level of the quantum divide-and-conquer algorithm given in \cite{AJ22}, using the fault-tolerant quantum minimum finding in \cite{WY23}.
Even though the logarithmic improvements are small in each level, they can significantly add up to quasi-polylogarithmic improvements in the overall complexity.

\subsection{Main result}

Suppose $s$ is a non-empty ($0$-indexed) string of length $n$. For every $i \in \sbra*{n}$,\footnote{We use the notation $\sbra*{n} = \cbra*{0, 1, 2, \dots, n - 1}$ for every positive integer $n$.} let
\[
s^{\rbra*{i}} = s\substr{i}{n-1} s\substr{0}{i-1}
\]
denote the rotation of $s$ with start position at index $i$.
(The function version of) LMSR is formally defined as follows.

\begin{problem} [Function problem of LMSR] \label{prob:function}
    Given a string $s$ of length $n$, find the smallest index $k \in \sbra*{n}$ such that $s^{(k)} \leq s^{(i)}$ for all $i \in \sbra*{n}$.
\end{problem}

Inspired by the quantum algorithms given in \cite{WY23} and \cite{AJ22}, we obtain better quantum query complexity of LMSR by a more refined design of divide and conquer.

\begin{theorem} [\cref{thm:function} restated] \label{thm:main}
    The worst-case quantum query complexity of LMSR is $\sqrt{n} \cdot 2^{O\rbra*{\sqrt{\log n}}}$.
\end{theorem}

For the function problem of LMSR, our quantum algorithm achieves a quasi-polylogarithmic improvement over the result of $\sqrt{n} \cdot 2^{\rbra*{\log n}^{1/2+\varepsilon}}$ by \cite{AJ22}.
For comparison, we collect all known (both classical and quantum) query complexities of LMSR in Table \ref{tab:overview}.

\begin{table}[!htp]
\centering
\begin{threeparttable}
\caption{Query Complexity of LMSR.}
\label{tab:overview}
\begin{tabular}{ccccc}
\toprule
\begin{tabular}[c]{@{}c@{}}Complexity \\ Type \end{tabular} & \begin{tabular}[c]{@{}c@{}}Algorithm \\ Type \end{tabular} & \begin{tabular}[c]{@{}c@{}}Problem \\ Type \end{tabular} & \begin{tabular}[c]{@{}c@{}}Query \\ Complexity \end{tabular}                                               & References    \\ \midrule
Worst-Case      & Classical      & Function        & $\Theta\rbra*{n}$                                              & \cite{Boo80,Shi81,Duv83}             \\ \cmidrule{2-5}
                & Parallel       & Function        & $\Theta\rbra*{\log n}$ ${}^\dag$                                        & \cite{AIP87,IS92}             \\ \cmidrule{2-5}
                & Quantum        & Function        & $O\rbra*{n^{3/4}}$, $\Omega\rbra*{\sqrt{n}}$                   & \cite{WY23}    \\ \cmidrule{4-5}
                &                &         & $\sqrt{n} \cdot 2^{\rbra*{\log n}^{1/2+\varepsilon}}$ & \cite{AJ22}    \\ \cmidrule{4-5}
                &                &         & $\sqrt{n} \cdot 2^{O\rbra*{\sqrt{\log n}}}$ & This Work      \\ \cmidrule{3-5}
                &                & Decision        & $\widetilde O\rbra*{\sqrt{n}}$                                   & \cite{CKKD+22} \\ \midrule
Average-Case    & Classical      & Function        & $O\rbra*{n}$, $\Omega\rbra*{n/\log n}$                         & \cite{IS94,BCN05,WY23}   \\ \cmidrule{2-5}
                & Quantum        & Function        & $O\rbra*{\sqrt{n}\log n}$, $\Omega\rbra*{\sqrt{n/\log n}}$     & \cite{WY23}    \\ \bottomrule
\end{tabular}
\begin{tablenotes}
      \small
      \item ${}^\dag$ It is the time (depth) complexity of the parallel algorithm.
    \end{tablenotes}
\end{threeparttable}
\end{table}

Theorem \ref{thm:main} also implies that the quantum query complexity of minimal suffix, maximal suffix, and maximal string rotation are also improved to the same complexity as in Theorem \ref{thm:main} because they can be easily reduced to each other as discussed in \cite{AJ22}.

\subsection{Techniques}

In the quantum divide-and-conquer algorithm for LMSR in \cite{AJ22}, they reduce the original LMSR problem of size $n$ to $b$ small LMSR problems of size $n/b$ (with $b$ a parameter to be determined that depends on $n$) with an exclusion rule (see \cref{thm:exclusion-rule}) improving the one in \cite{WY23}.
Specifically, they solve a more general the problem called minimal length-$\ell$ substrings, which is to find the (leftmost, in case of a tie) lexicographically minimal substring of length $\ell$ (see \cref{prob:length-l} for the formal definition).
It should be noted that the LMSR of a string $s$ is actually the minimal length-$n$ substring of $ss$ (the concatenation of two $s$'s).
Their approach for the problem of minimal length-$\ell$ substrings is as follows.
\begin{enumerate}
    \item For finding the minimal length-$\ell$ substring of a string $s$ of length $n$, let $m = \ceil{\ell/b} = O\rbra{n/b}$, then the problem can be solved from $O\rbra{n/m} = O\rbra{b}$ sub-problems, each of size $2m$ (for simplicity, we assume that $n/2 \leq \ell \leq n$ and $n$ is divisible by $m$):
    \begin{itemize}
        \item The $i$-th sub-problem ($i \geq 0$) is to find the (leftmost and rightmost) minimal length-$m$ substring of $s\substr{im}{\rbra{i+2}m-1}$ of length $2m$. 
        Let $x_i$ and $y_i$ be the indices of the leftmost and rightmost minimal length-$m$ substrings. 
    \end{itemize}
    \item According to the exclusion rule (see \cref{thm:exclusion-rule}), an occurrence of the minimal length-$\ell$ substring should start at an index among $x_i$'s and $y_i$'s over all $i \leq O\rbra{b}$.
    There are $O\rbra{b}$ candidates in total, and each two candidates can be compared with quantum query complexity $O\rbra{\sqrt{\ell}}$ by Grover search \cite{Gro96}. 
    The (lexicographically) minimal candidate will be an occurrence of the minimal length-$\ell$ substring, which can be found by quantum minimum finding \cite{DH96}.
    \item Finally, since an occurrence of the minimal length-$\ell$ substring is found, its leftmost occurrence can be found by the quantum string matching algorithm \cite{RV03} with quantum query complexity $\widetilde O\rbra{\sqrt{n}}$. 
\end{enumerate}
Now we analyze the quantum computational complexity of the above algorithm. 
Let $T\rbra{n}$ be the quantum query complexity for finding the minimal length-$n$ substring of a string of length $n$.
Then, each sub-problem can be solved with quantum query complexity $T\rbra{2m}$. 
Note that the answer can be found by quantum minimum finding \cite{DH96} among $O\rbra{b}$ candidates and the comparator between any two candidates can be implemented with success probability at least $2/3$ with quantum query complexity $T\rbra{2m} + O\rbra{\sqrt{\ell}}$. 
Finally, it turns out that the quantum query complexity for LMSR satisfies the recurrence relation:
\begin{align*}
    T\rbra{n} \leq O\rbra*{\sqrt{b}} \cdot \rbra*{T\rbra{2m} + O\rbra*{\sqrt{\ell}}} \cdot O\rbra{\log b} + \widetilde O\rbra{\sqrt{n}},
\end{align*}
where the factor $O\rbra{\log b}$ comes from majority voting to ensure that the success probability of the comparator is at least $1 - 1/\poly\rbra{b}$.
This recurrence relation simplifies to
\[
    T\rbra*{n} \leq \widetilde O\rbra*{\sqrt{b}} \cdot \rbra*{ T\rbra*{\frac{n}{b}} + O\rbra*{\sqrt{n}} } + \widetilde O\rbra*{\sqrt{n}},
\]
which gives the quantum query complexity
\[
T\rbra*{n} = \sqrt{n} \cdot 2^{\rbra*{\log n}^{1/2+\varepsilon}}
\]
for any constant $\varepsilon > 0$.
In the approach of \cite{AJ22}, they did not seek to optimize the polylogarithmic factors in each level of the divide-and-conquer process.
However, for the function version of LMSR, optimizing polylogarithmic factors will make a significant difference shown as follows.

In our approach, we try to optimize polylogarithmic factors in each level of the divide and conquer by appropriately using the quantum minimum finding algorithm on bounded-error oracles \cite{WY23} (which extends the result of \cite{HMdW03} on quantum search). As a result, we remove the $O\rbra{\log b}$ factor from the previous recurrence relation and thus obtain a new recurrence relation of the form
\[
    T\rbra*{n} \leq O\rbra*{\sqrt{b}} \cdot \rbra*{ T\rbra*{\frac{n}{b}} + O\rbra*{\sqrt{n}} } + \widetilde O\rbra*{\sqrt{n}},
\]
which, surprisingly, gives a better quantum query complexity
\[
T\rbra*{n} = \sqrt{n} \cdot 2^{O\rbra*{\sqrt{\log n}}}.
\]
This is a quasi-polylogarithmic improvement over the quantum query complexity $\sqrt{n} \cdot 2^{\rbra*{\log n}^{1/2+\varepsilon}}$ given in \cite{AJ22}.

\subsection{Applications}

As an application, we provide an efficient quantum algorithm for benzenoid identification.
This is a natural problem in organic chemistry, with a linear-time classical algorithm proposed in \cite{Bas16} using the boundary-edges code (BEC) \cite{HLZ96}. 
Although quantum speedups for benzenoid identification have been considered in \cite{WY23}, a better quantum query complexity can be obtained by the quantum algorithm proposed in this paper. 

Formally, a benzenoid can be described by a BEC, i.e., a finite string over $\cbra{1, 2, 3, 4, 5, 6}$. 
For any BEC $s$ of a benzenoid, the canonical BEC of a benzenoid is the lexicographically maximal string among all rotations of $s$ and $s^{\mathrm{R}}$ (the reverse of $s$), with the lexicographical order $1 < 2 < 3 < 4 < 5 < 6$. 
Our quantum algorithm in \cref{thm:main} can be used to find the canonical BEC of a benzenoid of size $n$ with quantum query complexity $\sqrt{n} \cdot 2^{O\rbra{\sqrt{\log n}}}$.
This is done as follows, with the same idea as in \cite{WY23}. 
\begin{enumerate}
    \item Let $i$ and $i^{\mathrm{R}}$ be the LMSR indices of $s$ and $s^{\mathrm{R}}$ (as defined in \cref{prob:function}), respectively, where the LMSR is defined with the lexicographical order $6 < 5 < 4 < 3 < 2 < 1$.
    This can be done with quantum query complexity $\sqrt{n} \cdot 2^{O\rbra{\sqrt{\log n}}}$ by \cref{thm:main}.
    \item The canonical BEC is either $s^{\rbra{i}}$ or $\rbra{s^{\mathrm{R}}}^{\rbra{i^{\mathrm{R}}}}$, which can be determined by the Grover search \cite{Gro96} with quantum query complexity $O\rbra{\sqrt{n}}$. 
\end{enumerate}
Finally, it can be seen that the canonical BEC of a benzenoid of size $n$ can be determined with quantum query complexity $\sqrt{n} \cdot 2^{O\rbra{\sqrt{\log n}}} + O\rbra{\sqrt{n}} = \sqrt{n} \cdot 2^{O\rbra{\sqrt{\log n}}}$, which improves the prior quantum approaches with quantum query complexity $O\rbra{n^{3/4}}$ shown in \cite{WY23} and $\sqrt{n}\cdot 2^{\rbra{\log n}^{1+\varepsilon}}$ implied in \cite{AJ22}. 

\subsection{Discussion} \label{sec:discussion}

In this note, we obtained a better quantum query complexity for (the function problem of) LMSR.
We hope it would bring new inspirations to discover and improve more quantum divide-and-conquer algorithms, especially for quantum algorithms for string processing \cite{CILG+12,AM14,Mon17,BEG+21,LGS22}.

The observation of this paper leads to an improved recurrence relation for quantum divide and conquer, removing the polylogarithmic factors in certain cases (see \cref{sec:qdc} for illustration).
We summarize the main point of our results below.
\begin{itemize}
    \item \textit{A baby step is a giant step}: Small factors are considerable in (quantum) divide and conquer. Sometimes an improvement in each level of the divide and conquer will yield a significant speedup, even beyond polynomial (with respect to the improvement made in each level). Our quantum algorithm for the function problem of LMSR makes a polylogarithmic improvement in each level of the divide and conquer, and finally yields a quasi-polynomial speedup with respect to the improvement, resulting in a quasi-polylogarithmic speedup with respect to the size of the problem, in the overall query complexity. An interesting question is: can we obtain polynomial and even exponential speedups (with respect to the size of the problem) by quantum divide and conquer in certain problems?
\end{itemize}

As can be seen in our result of quantum query complexity, there is still a quasi-polylogarithmic gap between the function and decision problems of LMSR. The reason is that the divide-and-conquer tool given in \cite{CKKD+22} does not apply to non-Boolean cases (see \cref{remark:non-Boolean} for more discussions). If we could obtain a similar identity for adversary composition in the general case, then we might improve the quantum query complexity of the function problem of LMSR and other problems that can be solved by quantum divide and conquer.

\subsection{Organization of this paper}

In the rest of this note, we first introduce the necessary preliminaries in \cref{sec:preliminary}. Then, in \cref{sec:function}, we will study the quantum query complexity of LMSR.

\section{Preliminaries} \label{sec:preliminary}

In this section, we introduce the notations of strings, the concept of quantum query complexity, and basic quantum subroutines such as quantum minimum finding and quantum string matching.

\subsection{Strings}

Let $\Sigma$ be a finite alphabet with a total order $<$. A string $s \in \Sigma^n$ of length $\abs*{s} = n$ is a function $s \colon \sbra*{n} \to \Sigma$. We write $s\sbra*{i}$ to denote the $i$-th character of $s$ for $i \in \sbra*{n}$, and $s\substr{i}{j}$ to denote the substring $s\sbra*{i} s\sbra*{i+1} \dots s\sbra*{j}$ for $0 \leq i \leq j < n$.
Two strings $s$ and $t$ are equal, denoted by $s = t$, if $\abs*{s} = \abs*{t}$ and $s\sbra*{i} = t\sbra*{i}$ for every $i \in \sbra*{\abs*{s}}$; $s$ is lexicographically smaller than $t$, denoted by $s < t$, if either $s$ is a proper prefix of $t$ or there is an index $i \in \sbra*{\min\cbra*{\abs*{s}, \abs*{t}}}$ such that $s\sbra*{i} < t\sbra*{i}$. We write $s \leq t$ if $s = t$ or $s < t$.

\subsection{Quantum query complexity}

We assume quantum access to the input string $s$ of length $n$. More precisely, there is a quantum unitary oracle $O_s$ such that
\[
    O_s \colon \ket*{i}\ket*{0} \mapsto \ket*{i}\ket*{s\sbra*{i}}
\]
for every $i \in \sbra*{n}$. A quantum query algorithm $A$ with $T$ queries to quantum oracle $O_s$ is a sequence of unitary operators
\[
    A \colon U_0 \to O_s \to U_1 \to O_s \to \dots \to O_s \to U_T,
\]
where $U_0, U_1, \dots, U_T$ are uniform quantum unitary operators that are determined independent of $O_s$ (but depend on $n$). Let $\cbra*{P_y}$ be a quantum measurement in the computational basis. The probability that $A$ outputs $y$ on input $s$ is defined by
\[
    \Pr \sbra*{ A\rbra*{s} = y } = \Abs*{P_y U_T O_s \dots O_s U_1 O_s U_0 \ket{0}}^2.
\]
The quantum query complexity of computational task $f$, denoted by $Q\rbra*{f}$, is a function of $n$ that is the minimal possible number of queries to $O_s$ used in a quantum query algorithm $A$ such that $\Pr \sbra*{ A\rbra*{s} = f\rbra*{s} } \geq 2/3$ for every input $s$ of length $n$.

\subsection{Quantum minimum finding}

Quantum minimum finding \cite{DH96,AK99,DHHM06} is a basic subroutine of query complexity $O\rbra*{\sqrt{n}}$ widely used in quantum algorithms. However, when the quantum oracle is bounded-error, we need $O\rbra*{\log n}$ primitive queries to reduce the error probability for one logical query, which results in a quantum algorithm for minimum finding with query complexity $O\rbra*{\sqrt{n} \log n}$.
Specifically, unitary operator $O_{\mathrm{cmp}}$ is said to be a bounded-error oracle with respect to comparator $\mathrm{cmp} \colon \sbra*{n} \times \sbra*{n} \to \cbra*{0, 1}$, if
\[
    O_{\mathrm{cmp}} \ket*{i}\ket*{j}\ket*{0}\ket*{0}_w = \sqrt{p_{ij}} \ket*{i}\ket*{j}\ket*{\mathrm{cmp}\rbra*{i, j}}\ket*{\phi_{ij}}_w + \sqrt{1-p_{ij}} \ket*{i}\ket*{j}\ket*{1-\mathrm{cmp}\rbra*{i, j}}\ket*{\varphi_{ij}}_w,
\]
where $p_{ij} \geq 2/3$ for every $i \in \sbra*{n}$, and $\ket*{\phi_{ij}}_w$ and $\ket*{\varphi_{ij}}_w$ are ignorable work qubits for every $i, j \in \sbra*{n}$.
Here, the comparator $\mathrm{cmp}$ is assumed to induce a strict total order ``$<_{\mathrm{cmp}}$'' such that $i <_{\mathrm{cmp}} j$ if and only if $\mathrm{cmp}\rbra{i, j} = 1$.
Intuitively, bounded-error quantum oracles are understood as a quantum generalization of probabilistic bounded-error oracles that return the correct answer with probability at least $2/3$.
Inspired by the error reduction for quantum search in \cite{HMdW03}, it was shown in \cite{WY23} that quantum minimum finding can also be solved with query complexity $O\rbra*{\sqrt{n}}$, even if the quantum oracle is bounded-error. Specifically, we can find the minimum element with probability $\geq 2/3$ using $O\rbra*{\sqrt{n}}$ queries to $O_{\mathrm{cmp}}$. We formally state the result as follows.

\begin{lemma} [Quantum minimum finding on bounded-error oracles, {\cite[Lemma 3.4]{WY23}}] \label{lemma:qmin-bounded-error}
    There is a quantum algorithm for minimum finding on bounded-error oracles with query complexity $O\rbra*{\sqrt{n}}$.
\end{lemma}


\subsection{Quantum divide and conquer} \label{sec:qdc}

Divide and conquer is an important and basic idea in the design of (classical) algorithms. 
Recently, it was shown to be useful in quantum computing with a concrete example in string problems \cite{AJ22}. 
The complexity of quantum divide and conquer was later investigated in \cite{CKKD+22} and \cite{ABB+23}. 
Here, we briefly introduce the idea of quantum divide and conquer. 

Suppose a problem of size $n$ can be solved by combining the answers to several smaller problems of the same structure, say $a$ sub-problems of size $n/b$. 
Let $T\rbra{n}$ be the quantum query complexity of the problem of size $n$. 
If the answer to the problem of size $n$ can be obtained by taking the minimum (or maximum) answer over all the $a$ sub-problems of size $n/b$ with postprocessing complexity $f\rbra{n}$, then the following recurrence relation holds:
\[
    T\rbra{n} \leq \widetilde O\rbra*{\sqrt{a}} \cdot T\rbra*{\frac{n}{b}} + f\rbra{n},
\]
where the factor $\widetilde O\rbra{\sqrt{a}}$ is achieved by applying the quantum minimum finding \cite{DH96}. 

It is worth noting that, using the observation of this paper, the recurrence relation can be further improved to 
\[
    T\rbra{n} \leq O\rbra*{\sqrt{a}} \cdot T\rbra*{\frac{n}{b}} + f\rbra{n}
\]
by the quantum minimum finding on bounded-error oracles (see \cref{lemma:qmin-bounded-error}), removing the polylogarithmic factors from $\widetilde O\rbra{\sqrt{a}}$. 
It turns out that this improvement in the recurrence relation can further lead to quasi-polylogarithmic improvements in the quantum query complexity of certain problems (see \cref{thm:function}), which is the main contribution of this paper. 

\begin{remark} \label{remark:non-Boolean}
    A similar improvement in the recurrence relation of quantum divide and conquer was achieved in \cite{CKKD+22} for the Boolean case where the answer to each (sub-)problem is either $0$ or $1$. 
    In this case, taking the minimum (resp.\ maximum) degenerates to the Boolean AND (resp.\ OR) operation. 
    However, it is not clear how to generalize their results to non-Boolean cases. 
\end{remark}

\subsection{Basic quantum subroutines for string problems} \label{sec:quantum-string-matching}

Comparing two strings in the lexicographical order is a basic problem in stringology, which can be done by directly applying Grover search \cite{Gro96}, quadratically better than any classical string comparators.
For completeness, we state a simple quantum subroutine for string comparator (see, for example, \cite[Lemma 4.1]{WY23}).

\begin{lemma} \label{lemma:cmp}
    Given two strings $s$ and $t$, there is a quantum algorithm that determines whether $s < t$ or not with query complexity $O\rbra{\sqrt{\min\cbra{\abs{s}, \abs{t}}}}$, with success probability $\geq 2/3$.
\end{lemma}

Next, we introduce the quantum string matching algorithm \cite{RV03} and its improved version \cite{WY23}, which is based on deterministic sampling \cite{Vis91}.

\begin{lemma} [Quantum string matching, {\cite[Section 4.3.2]{WY23}}] \label{lemma:quantum-string-matching}
    Given text $t$ of length $n$ and pattern $p$ of length $m$, there is a quantum algorithm to finds the leftmost and rightmost occurrences of $p$ in $t$ or reports that $p$ does not occur in $t$ with query complexity $O\rbra*{\sqrt{n\log m} + \sqrt{m \log^3 m \log \log m}}$, with success probability $\geq 2/3$.
\end{lemma}

\section{Quantum Query Complexity of LMSR} \label{sec:function}

In this section, we study the quantum query complexity of the function problem of LMSR.

\subsection{Reduction to minimal length-\texorpdfstring{$\ell$}{l} substrings}

As in \cite{AJ22}, LMSR is reduced to the problem of minimal length-$\ell$ substrings, which is defined as follows.

\begin{problem} [Minimal length-$\ell$ substrings] \label{prob:length-l}
    Given a string $s$ of length $n$, find the smallest index $k \in \sbra*{n-\ell}$ such that $s\substr{k}{k+\ell-1} \leq s\substr{i}{i+\ell-1}$ for all $i \in \sbra*{n-\ell}$.
\end{problem}

\begin{remark}
    The problem of minimal length-$\ell$ substrings defined in Problem \ref{prob:length-l} is slightly different from that in \cite{AJ22}. For our purpose, we only have to focus on the leftmost occurrence of the minimal length-$\ell$ substring.
\end{remark}

Here, we recall the exclusion rule for minimal length-$\ell$ substrings used in \cite{AJ22}.

\begin{theorem} [Exclusion rule for minimal length-$\ell$ substrings, {\cite[Lemma 4.8]{AJ22}}] \label{thm:exclusion-rule}
    Suppose $s$ is a string of length $n$, and $n/2 \leq \ell \leq n$. Let
    \[
        I = \set{ k \in \sbra*{n-\ell+1} }{ s\substr{k}{k+\ell-1} = \min_{i \in \sbra*{n-\ell}} s\substr{i}{i+\ell-1} }
    \]
    be the set of all indices of minimal length-$\ell$ substrings of $s$, which form an arithmetic progression. For every integer $a \geq 0$ and $m \geq 1$ with $a + m \leq n - \ell$, let
    \[
        J = \set{ a \leq k < a + m }{ s\substr{k}{k+m-1} = \min_{a \leq i < a + m} s\substr{i}{i+m-1} }
    \]
    be the set of all indices of minimal length-$m$ substrings of $s\substr{a}{a+2m-1}$. If $I \cap \cbra*{\min J, \max J} = \emptyset$, then $I \cap J = \emptyset$.
\end{theorem}

For simplicity as well as completeness, we give a direct reduction from LMSR to minimal length-$\ell$ substrings as follows.

\begin{proposition} \label{prop:LMSR-to-minimal-length-l}
    LMSR can be reduced to the problem of minimal length-$\ell$ substrings defined in Problem \ref{prob:length-l}.
\end{proposition}

\begin{proof}
    Suppose string $s$ of length $n$ is given, and we want to find the LMSR of $s$.
    Let $s' = ss$ be the concatenation of two $s$'s.
    Then it can be shown that the LMSR of $s$ is the leftmost minimal length-$n$ substring of $s'$ of length $2n$.
\end{proof}

\subsection{Improved quantum query complexity}

We provide a quantum algorithm for the problem of minimal length-$\ell$ substrings with better quantum query complexity as follows.

\begin{theorem} [Quantum query complexity of LMSR] \label{thm:function}
    The worst-case quantum query complexity of Problem \ref{prob:function} is $\sqrt{n} \cdot 2^{O\rbra*{\sqrt{\log n}}}$.
\end{theorem}

\begin{proof}
    The framework of our algorithm follows from that of {\cite[Theorem 4.1]{AJ22}} with improvements.
    By \cref{prop:LMSR-to-minimal-length-l}, the LMSR of a string $s$ can be reduced to finding the leftmost minimal length-$n$ substring of $s' = ss$.
    For simplicity, we consider how to find the leftmost minimal length-$\ell$ substring of a string $s$ of length $n$ with $n/2 \leq \ell \leq n$.
    Suppose $s$ is a string of length $n$, and $O_s$ is the quantum oracle with access to $s$.
    For convenience, we also find the rightmost minimal length-$\ell$ substring of $s$.
    Let $\mathsf{solve}\rbra*{s, \ell}$ be a function that returns two indices $x$ and $y$ indicating the leftmost and rightmost indices of the minimal length-$\ell$ substrings, denoted as $\rbra{x, y} \gets \mathsf{solve}\rbra*{s, \ell}$.

    \textbf{Dividing into sub-problems}.
    Let $b \geq 1$ be some parameter to be determined that depends on $n$, and $m = \ceil*{\ell/b}$. We will recursively reduce the problem of minimal length-$\ell$ substrings of size $n$ to $\ceil*{n/m}-1$ sub-problems of size at most $2m$.

    To this end, for every $i \in \sbra*{\ceil*{n/m}-1}$, the string of the $i$-th sub-problem is
    \[
    s_i = s\substr{im}{\min\cbra*{\rbra*{i+2}m-1, n - 1}},
    \]
    and we recursively solve the problem of minimal length-$\ell_i$ substrings with $\ell_i = m$.
    Suppose $x_i$ (resp. $y_i$) is the index of the leftmost (resp. rightmost) minimal length-$\ell_i$ substring of $s_i$, where $im \leq x_i \leq y_i \leq \min\cbra*{\rbra*{i+2}m-1, n - 1}$. Note that
    \[
        s\substr{x_i}{x_i + m - 1} = s\substr{y_i}{y_i + m - 1} = \min_{im \leq j \leq \min\cbra*{\rbra*{i+2}m-1, n - 1}} s\substr{j}{j+m-1}.
    \]
    By Theorem \ref{thm:exclusion-rule}, we know that the index of the minimal length-$\ell$ substring of $s$ exists among the indices $x_i$ and $y_i$. Therefore, we only have to find the minimal length-$\ell$ substring among the $2\rbra*{\ceil*{n/m}-1}$ indices $x_i$ and $y_i$. To achieve this, we construct a comparator $\mathrm{cmp} \colon \sbra*{\ceil*{n/m}-1} \times \sbra*{\ceil*{n/m}-1} \to \cbra*{0, 1}$ such that
    \[
        \mathrm{cmp}\rbra*{i, j} = \begin{cases}
            1, & t_i < t_j \lor \rbra{t_i = t_j \land i < j}, \\
            0, & \text{otherwise},
        \end{cases}
    \]
    where $t_i = \min\cbra*{s\substr{x_i}{x_i+\ell-1}, s\substr{y_i}{y_i+\ell-1}}$.
    It can be verified that $\mathrm{cmp}$ induces a strict total order ``$<_{\mathrm{cmp}}$''.

    \textbf{Solving minimal length-$\ell$ substrings with queries to $O_{\mathrm{cmp}}$}.
    Assume that there is a bounded-error quantum oracle $O_{\mathrm{cmp}}$ with respect to the comparator $\mathrm{cmp}$.
    Then, by \cref{lemma:qmin-bounded-error}, the leftmost minimal length-$\ell$ substring of $s$ is the minimum index among $\sbra*{\ceil*{n/m}-1}$ under the strict total order ``$<_{\mathrm{cmp}}$'' induced by the comparator $\mathrm{cmp}$, which can be found using $O\rbra{\sqrt{n/m}}$ queries to $O_{\mathrm{cmp}}$, with success probability $\geq 0.99$.
    Using $O_{\mathrm{cmp}}$, we can find the leftmost minimal length-$\ell$ substring as follows.
    \begin{enumerate}
        \item Find the minimum index $z \in \sbra*{\ceil*{n/m}-1}$ under the strict total order ``$<_{\mathrm{cmp}}$'' by \cref{lemma:qmin-bounded-error}, with success probability $\geq 0.99$, using $O\rbra{\sqrt{n/m}}$ queries to $O_{\mathrm{cmp}}$.
        \item Find the leftmost occurrence $x$ and the rightmost occurrence $y$ of $s\substr{z}{z+\ell-1}$ in $s$ by \cref{lemma:quantum-string-matching}, with success probability $\geq 0.99$, using $O\rbra*{\sqrt{n \log \ell} + \sqrt{\ell \log^3 \ell \log \log \ell}}$ queries to $O_{s}$.
    \end{enumerate}

    Let $T^{\mathsf{solve}}\rbra{n}$ be the quantum query complexity for finding the minimal length-$\ell$ substrings of a string of length $n$, and let $T^{\mathrm{cmp}}\rbra{n}$ be the quantum query complexity of $O_{\mathrm{cmp}}$.
    Both complexities are measured by the number of queries to the quantum oracle $O_{s}$.
    Then, according to the above simple procedure, we have
    \begin{equation} \label{eq:def-T-solve}
    T^{\mathsf{solve}}\rbra{n} = O\rbra*{\sqrt{\frac{n}{m}}} \cdot T^{\mathrm{cmp}}\rbra{n} + O\rbra*{\sqrt{n \log \ell} + \sqrt{\ell \log^3 \ell \log \log \ell}}.
    \end{equation}

    \textbf{Implementation of $O_{\mathrm{cmp}}$}.
    Next, we show how to implement a bounded-error quantum oracle $O_{\mathrm{cmp}}$ with respect to the comparator $\mathrm{cmp}$, given a (quantum) algorithm for computing $\mathsf{solve}\rbra*{s_i, \ell_i}$ with bounded error (say, with success pobability $\geq 0.99$).
    This is done through three steps as follows.
    \begin{enumerate}
        \item Let $\rbra{x_i, y_i} \gets \mathsf{solve}\rbra*{s_i, \ell_i}$ with success probability $\geq 0.99$. By induction, this can be done by using $T^{\mathsf{solve}}\rbra{\abs{s_i}}$ queries to $O_s$.
        \item Let $\rbra{x_j, y_j} \gets \mathsf{solve}\rbra*{s_j, \ell_j}$ with success probability $\geq 0.99$. By induction, this can be done by using $T^{\mathsf{solve}}\rbra{\abs{s_j}}$ queries to $O_s$.
        \item Check if $t_i < t_j$ or $t_i = t_j \land i < j$ according to $x_i, y_i, x_j, y_j$ by \cref{lemma:cmp} with success probability $\geq 0.99$.
        This can be done by using $O\rbra{\sqrt{\min\cbra{\abs{t_i}, \abs{t_j}}}} = O\rbra{\sqrt{\ell}}$ queries to $O_s$.
    \end{enumerate}
    It can be seen that the above process will compare two indices with respect to $<_{\mathrm{cmp}}$ with success probability $\geq 0.99^3 \geq 0.9$, and the quantum query complexity of $O_{\mathrm{cmp}}$ is
    \begin{align}
    T^{\mathrm{cmp}}\rbra*{n}
    & = T^{\mathsf{solve}}\rbra*{\abs{s_i}} + T^{\mathsf{solve}}\rbra*{\abs{s_j}} + O\rbra*{\sqrt{\ell}} \nonumber \\
    & \leq T^{\mathsf{solve}}\rbra*{2m} + T^{\mathsf{solve}}\rbra*{2m} + O\rbra*{\sqrt{\ell}} \nonumber \\
    & = 2 T^{\mathsf{solve}}\rbra*{2m} + O\rbra*{\sqrt{\ell}}. \label{eq:def-T-cmp}
    \end{align}

    \textbf{Complexity analysis}.
    By \cref{eq:def-T-solve} and \cref{eq:def-T-cmp}, we have
    \[
        T^{\mathsf{solve}}\rbra{n} \leq O\rbra*{\sqrt{\frac{n}{m}}} \cdot \rbra*{2 T^{\mathsf{solve}}\rbra*{2m} + O\rbra*{\sqrt{\ell}}} + O\rbra*{\sqrt{n \log \ell} + \sqrt{\ell \log^3 \ell \log \log \ell}}.
    \]
    Since $n/2 \leq \ell \leq n$ and $m = \ceil{\ell/b}$, we have
    \begin{equation} \label{eq:T-solve-with-b}
    T^{\mathsf{solve}}\rbra{n} \leq O\rbra*{\sqrt{b}} \cdot \rbra*{T^{\mathsf{solve}}\rbra*{\ceil*{\frac{2n}{b}}} + O\rbra*{\sqrt{n}}} + O\rbra*{\sqrt{n \log^3 n \log \log n}}.
    \end{equation}
    Using $b' = b/4$ and with the assumption that $1 \leq b \leq n/2$,
    \cref{eq:T-solve-with-b} becomes
    \begin{align*}
        T^{\mathsf{solve}}\rbra{n}
        & \leq O\rbra*{\sqrt{4b'}} \cdot \rbra*{T^{\mathsf{solve}}\rbra*{\frac{4n}{b}} + O\rbra*{\sqrt{n}}} + O\rbra*{\sqrt{n \log^3 n \log \log n}} \\
        & = O\rbra*{\sqrt{b'}} \cdot \rbra*{T^{\mathsf{solve}}\rbra*{\frac{n}{b'}} + O\rbra*{\sqrt{n}}} + O\rbra*{\sqrt{n \log^3 n \log \log n}}.
    \end{align*}
    Therefore, an upper bound on $T^{\mathsf{solve}}\rbra{n}$ can be given by the solution to the following recurrence relation of the form
    \begin{equation} \label{eq:recurrence-main}
        T\rbra{n} \leq c \cdot \rbra*{\sqrt{b} T\rbra*{\frac{n}{b}} + \sqrt{bn} + \sqrt{n \log^3 n \log \log n}},
    \end{equation}
    where $c > 1$ is a constant.

    \textbf{Solving the recurrence relation}.
    In the following, we write $\log n = \log_2 n$ for convenience.
    Let $d = 2 \sqrt{\log c} > 0$ and choose
    \begin{equation} \label{eq:def-b}
        b\rbra*{n} = 2^{d\sqrt{\log n}} = 2^{O\rbra*{\sqrt{\log n}}} = \omega\rbra*{\poly\rbra*{\log n}}.
    \end{equation}
    We will show by induction that the solution to \cref{eq:recurrence-main} satisfies that
    \begin{equation} \label{eq:induction}
        T\rbra*{n} \leq \sqrt{n} b\rbra*{n}
    \end{equation}
    for sufficiently large $n$. For sufficiently large $n$ such that $\log^3 n \log \log n \leq b\rbra*{n} \leq n/2$, the recurrence in \cref{eq:recurrence-main} becomes
    \begin{align*}
        T\rbra*{n}
        & \leq c \cdot \rbra*{ \sqrt{b\rbra*{n}} T\rbra*{\frac{n}{b\rbra*{n}}} + \sqrt{n b\rbra*{n}} + \sqrt{n \log^3 n \log \log n} } \\
        & \leq c \cdot \rbra*{ \sqrt{b\rbra*{n}} T\rbra*{\frac{n}{b\rbra*{n}}} + 2\sqrt{n b\rbra*{n}} }. \label{eq:T-simplified}
    \end{align*}
    By induction on $n$ in \cref{eq:induction}, it holds that $T\rbra{n/b\rbra{n}} \leq \sqrt{n/b\rbra{n}} b\rbra{n/b\rbra{n}}$, and we have
    \begin{align*}
        T\rbra*{n}
        & \leq c \cdot \rbra*{ \sqrt{b\rbra*{n}} \cdot \sqrt{\frac{n}{b\rbra*{n}}} b\rbra*{\frac{n}{b\rbra*{n}}} + 2\sqrt{n b\rbra*{n}} } \\
        & = c \sqrt{n} \cdot \rbra*{ b\rbra*{\frac{n}{b\rbra*{n}}} + 2 \sqrt{b\rbra*{n}} }.
    \end{align*}
    To see that $T\rbra*{n} \leq \sqrt{n} b\rbra*{n}$ for sufficiently large $n$, it is sufficient to show that
    \begin{equation} \label{eq:recurrence-function}
        \lim_{n \to \infty} \frac{b\rbra*{\dfrac{n}{b\rbra*{n}}} + 2\sqrt{b\rbra*{n}}}{b\rbra*{n}} < \frac 1 c.
    \end{equation}
    By direct calculation based on \cref{eq:def-b}, we have the identity
    \[
        b\rbra*{\frac n {b\rbra*{n}}} = 2^{d \sqrt{\log n - d \sqrt{\log n}}},
    \]
    and the limit
    \[
         \lim_{n \to \infty} \frac{2\sqrt{b\rbra*{n}}}{b\rbra*{n}} = 0.
    \]
    With these, the left hand side of Eq. (\ref{eq:recurrence-function}) becomes
    \begin{align*}
        \lim_{n \to \infty} \frac{ 2^{d \sqrt{\log n - d \sqrt{\log n}}} }{ 2^{d \sqrt{\log n} } }
        & = \lim_{n \to \infty} 2^{d \rbra*{ \sqrt{\log n - d \sqrt{\log n}} - \sqrt{\log n} }} \\
        & = \lim_{n \to \infty} 2^{d \frac{ -d \sqrt{\log n} }{ \sqrt{\log n - d \sqrt{\log n}} + \sqrt{\log n} } } \\
        & = \lim_{n \to \infty} 2^{ \frac {-d^2} {1 + \sqrt{1 - \frac{d}{\sqrt{\log n}}}} } \\
        & = 2^{- \frac{d^2} {2}} \\
        & = \frac 1 {c^2} \\
        & < \frac 1 c.
    \end{align*}
    Therefore, the solution to the recurrence is
    \[
        T\rbra*{n} \leq \sqrt{n} b\rbra*{n} = \sqrt{n} 2^{O\rbra*{\sqrt{\log n}}}.
    \]
\end{proof}

Compared to the approach in \cite{AJ22}, our speedup mainly comes from the appropriate use of the quantum minimum finding on bounded-error oracles (see Lemma \ref{lemma:qmin-bounded-error}). It can be seen that the quantum minimum finding in Lemma \ref{lemma:qmin-bounded-error} with query complexity $O\rbra*{\sqrt{n}}$ only achieves a logarithmic improvement over the na{\"{i}}ve method by error reduction with query complexity $O\rbra*{\sqrt{n}\log n}$. Surprisingly, it turns out that such $\poly\rbra*{\log n} = \widetilde O\rbra*{1}$ improvement on the frequently used subroutine can bring a $2^{\rbra*{\log n}^{\varepsilon}} = \omega\rbra*{\poly\rbra*{\log n}}$ speedup beyond polylogarithmic factors.

\section*{Acknowledgements}

The author would like to thank Ansis Rosmanis for pointing out an error in an earlier version of this paper, and thank Fran\c{c}ois Le Gall for helpful discussions. 

The work of Qisheng Wang was supported in part by the Engineering and Physical Sciences Research Council under Grant \mbox{EP/X026167/1} and in part by the MEXT Quantum Leap Flagship Program (MEXT Q-LEAP) under Grant \mbox{JPMXS0120319794}. 

\addcontentsline{toc}{section}{References}
\bibliographystyle{unsrturl}
\bibliography{main}

\end{document}